\documentclass[lettersize,journal]{IEEEtran}
\usepackage{amsmath,amsfonts}
\usepackage{algorithmic}
\usepackage{algorithm}
\usepackage{array}
\usepackage[caption=false,font=normalsize,labelfont=sf,textfont=sf]{subfig}
\usepackage{textcomp}
\usepackage{stfloats}
\usepackage{url}
\usepackage{verbatim}
\usepackage{graphicx}
\usepackage{cite}
\usepackage{color}
\usepackage{siunitx}
\usepackage{amssymb}
\usepackage{bm}
\usepackage{amsthm}
\newtheorem{assumption}{Assumption}
\newtheorem{corollary}{Corollary}
\begin{document}

\title{A Novel Framework of Horizontal-Vertical Hybrid Federated Learning for EdgeIoT} 

\author{Kai~Li,~\IEEEmembership{Senior Member,~IEEE,}
	Yilei~Liang,
	Xin~Yuan,~\IEEEmembership{Senior Member,~IEEE,}
	Wei~Ni,~\IEEEmembership{Fellow,~IEEE,}
	Jon~Crowcroft,~\IEEEmembership{Fellow,~IEEE,}
	Chau~Yuen,~\IEEEmembership{Fellow,~IEEE,}
        and Ozgur~B.~Akan,~\IEEEmembership{Fellow,~IEEE}
\thanks{K.~Li is with the Internet of Everything (IoE) Group, Department of Engineering, University of Cambridge, CB3 0FA Cambridge, UK, and also with Real-Time and Embedded Computing Systems Research Centre (CISTER), Porto 4249--015, Portugal (E-mail: kaili@ieee.org).}
\thanks{Y.~Liang and J.~Crowcroft are the Department of Computer Science and Technology, University of Cambridge, CB3 0FA Cambridge, UK (E-mail: yl841@cam.ac.uk \& Jon.Crowcroft@cl.cam.ac.uk).}
\thanks{X.~Yuan and W.~Ni are with Commonwealth Scientific and Industrial Research Organization (CSIRO), Sydney, NSW 2122, Australia, and the School of Computing Science and Engineering, University of New South Wales (E-mail: \{xin.yuan, wei.ni\}@data61.csiro.au).}
\thanks{C.~Yuen is with Nanyang Technological University, SG 639798, Singapore (E-mail: chau.yuen@ntu.edu.sg).}
\thanks{O. B. Akan is with the Internet of Everything (IoE) Group, Division of Electrical Engineering, Department of Engineering, University of Cambridge, CB3 0FA Cambridge, UK, and also with the Center for neXt-generation Communications (CXC), Ko\c c University, 34450 Istanbul, Turkey (E-mail: oba21@cam.ac.uk).}
\thanks{The source code of this work is released on GitHub: https://github.com/YukariSonz/VHFL\_SplitNN.}
}


\maketitle

\begin{abstract}
This letter puts forth a new hybrid horizontal-vertical federated learning (HoVeFL) for mobile edge computing-enabled Internet of Things (EdgeIoT). In this framework, certain EdgeIoT devices train local models using the same data samples but analyze disparate data features, while the others focus on the same features using non-independent and identically distributed (non-IID) data samples. Thus, even though the data features are consistent, the data samples vary across devices. The proposed HoVeFL formulates the training of local and global models to minimize the global loss function. Performance evaluations on CIFAR-10 and SVHN datasets reveal that the testing loss of HoVeFL with 12 horizontal FL devices and six vertical FL devices is 5.5\% and 25.2\% higher, respectively, compared to a setup with six horizontal FL devices and 12 vertical FL devices. 
\end{abstract}


\section{Introduction}
Federated learning (FL) presents a transformative solution to mobile edge computing-enabled Internet of Things (EdgeIoT). In an EdgeIoT-empowered hospital, each institution holds vast amounts of sensitive patient data collected from IoT devices, which can enhance the predictive capabilities of medical artificial intelligence (AI) models - critical for diagnoses, treatment planning, and outcomes prediction~\cite{li2022internet}. However, sharing patient data across institutions can pose privacy risks and violate regulations, such as The Health Insurance Portability and Accountability Act of 1996 in the United States or the Data Protection Act 2018 in the United Kingdom~\cite{ali2022federated}. 

EdgeIoT devices might have distinct sets of clients but collect unified data features on those patients, such as demographic information, medical histories, and treatment outcomes. Horizontal FL (HFL) can enable clinics to collaboratively train a machine learning model that captures the same features of non-independent and identically distributed (non-IID) data~\cite{park2022federated}. FL can also perform vertically, namely, vertical FL (VFL), where EdgeIoT trains on disparate types of information (i.e., data features) from the same clients~\cite{liu2022vertical}. 

In this letter, we propose a new hybrid horizontal-vertical FL (HoVeFL) for EdgeIoT, which enables some EdgeIoT devices to focus on the same features across non-IID samples and the others to focus on different features of the same samples. For instance, some EdgeIoT devices in HoVeFL analyze patient demographics and vital signs (e.g., age, blood pressure, and heart rate), while other devices analyze diagnostic images (e.g., X-rays and magnetic resonance imaging) from the same group of patients. Although the data samples (e.g., patients) are the same, the features (e.g., types of data) being analyzed are different across devices. Moreover, HoVeFL enables certain devices to concentrate on the same data features utilizing non-IID data samples. For example, some devices focus on patient vital signs (e.g., blood pressure and heart rate), but the data samples owned by each device are from different subsets of patients with varying conditions and demographics. Even though the features analyzed (vital signs) are the same, the actual data samples can be diverse and non-identically distributed across the devices.

HoVeFL has potential applications where data is rich in diversity and requires privacy, e.g., smart health systems with inherently heterogeneous devices and data sources. The EdgeIoT devices may vary widely in functionality. For example, wearable devices can monitor vitals such as heart rate and blood pressure, while diagnostic machines in hospitals might focus on more complex data like imaging scans~\cite{li2023towards}. In HoVeFL, some devices can engage in HFL with a focus on shared types of data (e.g., patient vitals) collected under varied conditions and demographics, addressing the non-IID nature of the data. Other devices can participate in VFL, analyzing different data types (e.g., imaging versus vitals) from the same patients. By utilizing HoVeFL, smart health systems can enhance medicine strategies of the EdgeIoT devices.

\section{Related Work}
A federated feature selection was studied in~\cite{zhang2023federated} to identify and discard irrelevant features using an enhanced one-class support vector machine. A feature relevance hierarchical clustering algorithm was presented to select overlapping features in HFL. In~\cite{arunan2023federated}, a feature similarity-matched parameter aggregation algorithm was studied to learn from diverse edge data. The algorithm examines the locally trained models from different sources and identifies neurons with probabilistically similar feature extraction functions. The authors of~\cite{you2022triple} focused on an informative device-activating HFL, which activates devices with abundant information to reduce communication costs. 

The authors of~\cite{shi2023vertical} introduced a VFL system that utilizes over-the-air computation to enhance model aggregation speed and accuracy, while reducing communication delays through aggregation among edge servers. The VFL combines transceiver and fronthaul quantization design, employing successive approximation to ensure efficient learning convergence. VFL was designed in~\cite{su2021hierarchical} to process hybrid data partitioning across IoT devices, each collecting unique type-specific features. A multitier-partitioned neural network architecture and a primal-dual transformation strategy were developed to decompose training across data samples and feature spaces. Under wireless communication constraints in VFL, enhancing model prediction accuracy while reducing system latency was studied in~\cite{yan2022latency} to address quantization and selective user participation in training. 

In~\cite{wang2020hybrid}, a framework employing hybrid differential privacy was developed for VFL. This framework constructs a generalized linear model using data that is vertically divided. The authors of~\cite{xu2021fedv} developed a hybrid framework to reduce communication overheads. This facilitates the development of precise models without the necessity for Taylor series approximations. In~\cite{lu2024two}, a two-dimensional hybrid FL framework was studied to address incomplete features. With an increasing number of dual-SIM-card users, VFL is employed to create a model for each pair of network service providers, integrating shareable features of users with dual SIM cards. 

The authors of~\cite{thapa2022splitfed} introduced a split learning method to improve data privacy and model robustness. The split learning combines HFL, which is parallel processing among distributed devices, and the split learning, which is network splitting into client-side and server-side sub-networks during training. In~\cite{liu2022wireless}, a split FL algorithm was developed to reap parallel model training and the model splitting structure of split learning. A multi-arm Bandit algorithm was studied to select user devices and the size of local model updates, balancing the channel qualities and the importance of local model updates.

\begin{figure}[htb]
\centering
\includegraphics[width=3.0in]{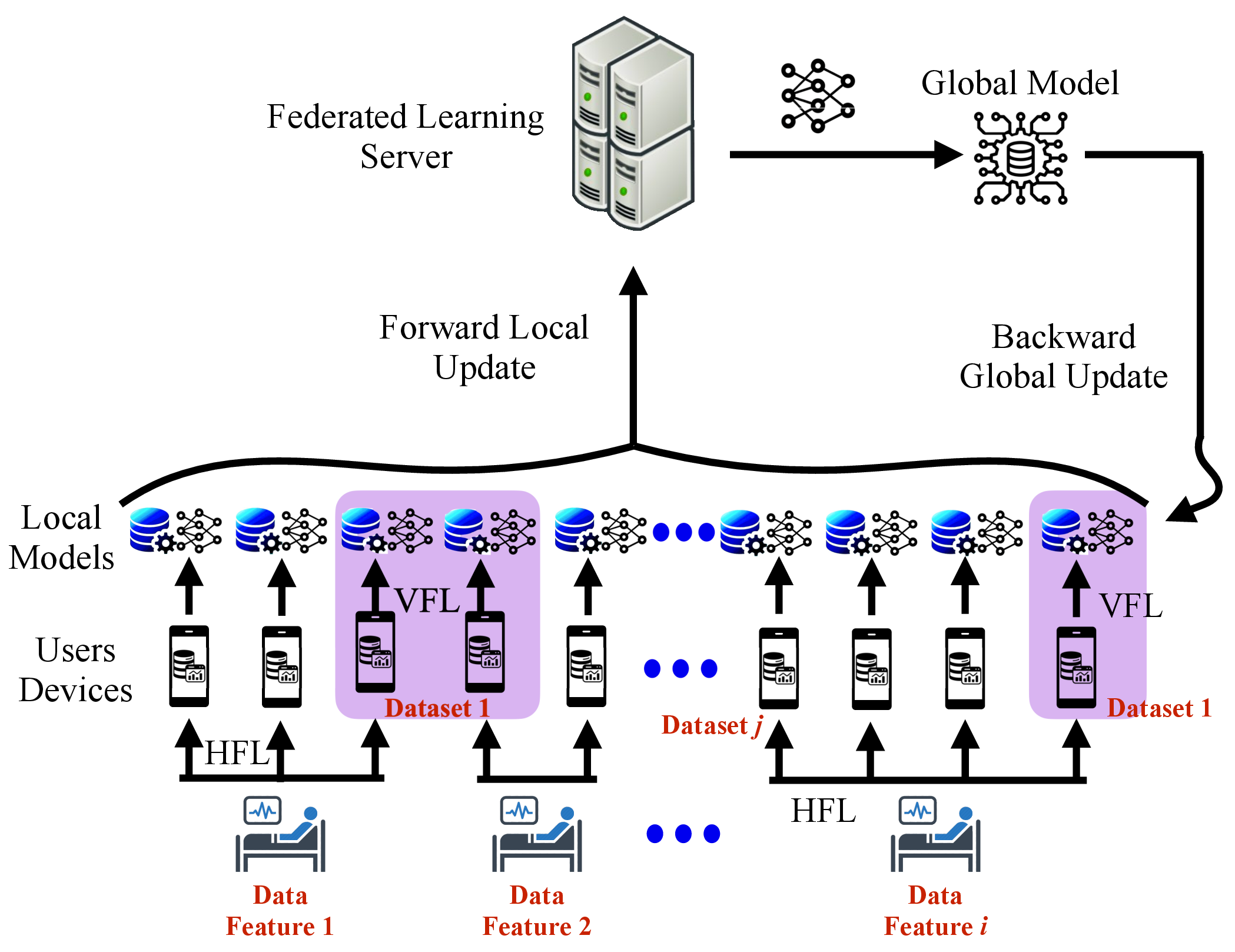}
\caption{A training process of the local and global models in HoVeFL, where each data feature $i$ ($i \in [1, I]$) is trained by $N_i$ devices. Moreover, each device has data samples $j$, thus, we have $n_{i,j} \in \bf{N_i} \cap \bf{N_j}$.} 
\label{fig_centralizedFL}
\end{figure}

The primary research gap addressed by the proposed HoVeFL framework lies in its integration of horizontal and vertical strategies within a single FL system. Existing models typically focus on either HFL, which processes shared data types across different populations, or VFL, which analyzes different data types for the same population but cannot leverage their respective strengths. HoVeFL bridges this gap by enabling simultaneous horizontal and vertical learning on heterogeneous data types and varied populations. This hybrid approach allows for more comprehensive model training that can adapt to non-IID data across devices, addressing the challenges of data diversity and distribution discrepancies more effectively than single HFL or VFL strategies. 

\section{Formulation of HoVeFL}
This section presents the training of the local and global models of HoVeFL in EdgeIoT for image classification as an example. Fig.~\ref{fig_centralizedFL} presents a HoVeFL training process with $N$ user devices, where data feature $i$ ($i \in [1, I]$) is trained by $N_i$ devices and data samples $j$ are trained by $N_j$ devices. Each device $n \in [1, N]$ has $n_{i,j}(t)$ with $j$ data samples and $i$ features at the $t$-th round of HoVeFL. Let $\pmb{m}^n_{i,j}(t) \in \mathbb{R}^{1\times M}$ denote the local model update in the $t$-th round. The training loss function of device $n$ is denoted by $f(\pmb{m}^n_{i,j}(t))$ and measures approximation errors based on training datasets in the $t$-th round. For example, the loss function can be modeled as linear or logistic regressions~\cite{li2024data}. 

Let $D_j^n(t)$ denote the data samples $j$ at device $n$. The local loss function of HoVeFL at device~$i$ for the $t$-th round is 
\begin{align}
F(\pmb{m}^n_{i,j}(t)) \!=\! \frac{1}{D_j^n(t)}\!\sum_{i=1}^{D_j^n(t)}& f\big(\pmb{m}^n_{i,j}(t) \big)\! +\! \alpha \zeta\big(\pmb{m}^n_{i,j}(t)\big), 
\label{eq_lossFunc}
\end{align}
where $\zeta(\cdot)$ is a regularizer function mitigating the effect of local training noise and potential overfitting due to overlapping input features among different devices. $\zeta(\cdot)$ contributes to the generalization ability of the model across the network. $\alpha \in [0,1]$ is a given coefficient.

Minimizing the local loss function in~\eqref{eq_lossFunc} leads to an optimal local model ${\pmb{m}^{n*}_{i,j}(t)}$ at device $n$, where
\begin{align}
{\pmb{m}^{n*}_{i,j}(t)} = \arg \min_{\forall \pmb{m}^n_{i,j}(t)} F(\pmb{m}^n_{i,j}(t)).
\label{eq_local_centralizedOptimal}
\end{align} 
With the learning rate $\mu_t$, the local model of device $i$ is updated using Stochastic Gradient Descent (SGD) for $T_L$ local iterations, throughout the $t$-th round: 
\begin{align}
\pmb{m}^n_{i,j}(t) \leftarrow \pmb{m}^n_{i,j}(t) - \mu_t \nabla F (\pmb{m}^n_{i,j}(t)).
\label{eq_local_SGD}
\end{align} 
\begin{figure*}
\begin{align}
F(\pmb{m}_G(t)) = \frac{1}{\sum_{i=1}^I N_i} \sum_{i=1}^I w^{n_i}_i(\pmb{m}^{1_i}_{i,1}(t) + \pmb{m}^{2_i}_{i,2}(t) + \cdots + \pmb{m}^{N_i}_{i,j}(t)) + 
\frac{1}{\sum_{j=1}^J N_j} \sum_{j=1}^J(w_1\pmb{m}^{1_j}_{1,j}(t) + w_2\pmb{m}^{2_j}_{2,j}(t) + \cdots + w^{n_j}_i\pmb{m}^{N_j}_{i,j}(t)). 
\label{eq_global_model}
\end{align} 
\hrulefill
\end{figure*}
After $T_L$ local iterations, the server aggregates $\pmb{m}^n_{i,j}(t), \forall i$ to obtain a global model, denoted by $\pmb{m}_G(t)$, for the $t$-th round. Then, $\pmb{m}_G(t)$ is broadcast to all devices for their training of $\pmb{m}^n_{i,j}(t+1),\,\forall i$ in the $(t+1)$-th round. The overlapping input features between the devices are handled differently. Although overlapping input features across devices are allowed, each device's participation in either HFL or VFL dictates its method of data integration. For VFL, where the devices might share partially overlapping input features on the same datasets, these features are aligned and integrated vertically to enrich the model's understanding of a device. This is achieved through techniques such as feature fusion or aggregation, ensuring that all relevant information from overlapping features is effectively utilized. For HFL, the emphasis is on leveraging the same set of features across different devices, and any overlapping features among the devices participating in HFL help enhance the generalizability and robustness of the model across diverse datasets~\cite{li2024leverage}. Thus, the overlap of input features is managed within the distinct protocols of either HFL or VFL, depending on the device's designated learning strategy.

In HoVeFL, $N_i$ users train the local models for the same data feature $i$ but on different data samples, namely, $\pmb{m}^1_{i,1}(t), \pmb{m}^2_{i,2}(t), \cdots, \pmb{m}^{N_i}_{i,j}(t)$, $i \in [1, I]$. Meanwhile, the other $N_j$ users train the local models based on similar sample $j$ but for different features, i.e., $\pmb{m}^1_{1,j}(t), \pmb{m}^2_{2,j}(t), \cdots, \pmb{m}^{N_j}_{i,j}(t)$, $j \in [1, J]$. Take a typical aggregation scheme, such as federated averaging (FedAvg)~\cite{zheng2022exploring}, as an example. At the server, a weighted model aggregation aggregates hybrid local models so that the global loss function $F(\pmb{m}_G(t))$ in~\eqref{eq_global_model} is minimized, where $w^n_i$ denotes the weight for each feature in the user's dataset as the server may apply a feature selection scheme. Although individual devices may not have a complete set of features to train independent models, they can contribute to a global model that enables independent predictions.

The aggregation in~\eqref{eq_global_model} can address the heterogeneity in the data and model updates. By employing advanced aggregation techniques, such as FedAvg with normalization factors and weighting schemes accounting for the diversity and size of the training data, we can align the updates cohesively. This approach ensures that the integration of VFL updates enhances the model by incorporating broad insights without overwhelming the HFL-type updates. Furthermore, these model updates can provide essential vertical insights that HFL alone cannot provide. By incorporating features from different domains, VFL updates enrich the global model with a more comprehensive understanding of the data relationships, which is critical for complex decision-making tasks.

\section{Convergence of HoVeFL}
In this section, we investigate the convergence of HoVeFL given non-IID data at the devices. 
\begin{assumption}\label{assumption}
	$\forall k \in {\cal K}$, we make the following assumptions:
	\begin{enumerate}
	\item $F(\pmb{m}_G(t))$ fulfills the Polyak-Lojasiewicz requirement~\cite{karimi2016linear} with a positive parameter $\rho$, indicating that $F(\pmb{m}_G(t))-F(\pmb{m}_G^*) \leq \frac{1}{2\rho} \left\| \nabla F(\pmb{m}_G(t)) \right\|^2 $ and ${\pmb{m}_G}^*$ minimizes $F(\pmb{m}_G(t))$, where $\| \cdot \|$ stands for norm operation;
	\item $F({\pmb{m}_G}(0)) - F({\pmb{m}_G(t)}^*) 
        =  \Theta$, where $\Theta$ is a constant.
	\end{enumerate}
\end{assumption}
Based on~\eqref{eq_lossFunc} and~\eqref{eq_global_model}, the global loss function minimized in HoVeFL can be defined as $\min_{\pmb{m}_G(t)} F(\pmb{m}_G(t))$. The gradient descent update is given by
\begin{align}
F(\pmb{m}_G(t \!+\! 1))\! =\! F(\pmb{m}_G(t)) \!-\! \mu_t \sum^{N}_{n=1} \nabla F(\pmb{m}_G(t)).
\label{eq_GlobalGradientUpdate}
\end{align}
Assume that the gradients of the loss functions $\nabla F(\pmb{m}_G(t))$ are $L$-Lipschitz continuous~\cite{elgabli2022fednew}. We have 
\begin{align}
\| \nabla F(\pmb{m}_G(t+1)) - &\nabla F(\pmb{m}_G(t)) \| \leq \nonumber \\
&L \| \pmb{m}_G(t+1) - \pmb{m}_G(t)\|.
\label{eq_Lipschitz}
\end{align}

Given the gradient descent update in~\eqref{eq_GlobalGradientUpdate} and the Lipschitz gradient condition in~\eqref{eq_Lipschitz}, an inequality that bounds the per-iteration change in the loss function is given by
\begin{align}
F&(\pmb{m}_G(t+1)) \leq F(\pmb{m}_G(t)) - \mu_t \Big( \nabla F(\pmb{m}_G(t)), \nonumber \\
&\sum_{n=1}^N \nabla F(\pmb{m}^n_{i,j}(t)) \Big) + \frac{L \mu_t^2}{2} \left\|\sum_{n=1}^N \nabla F(\pmb{m}^n_{i,j}(t))\right\|^2.
\label{eq_boundsLoss}
\end{align}

Furthermore, considering Cauchy-Schwarz and Jensen's inequality (which is used for handling expectations over convex functions) in the gradient descent update can smooth the model updates despite the heterogeneity in $D_j^n(t)$. In particular, the expectation of the convex combination of local losses (which are possibly based on different data distributions) is as good as the convex combination of their expectations, ensuring an averaging effect that is central to the convergence of HoVeFL. 

To derive explicit upper bounds of HoVeFL, we apply the Cauchy-Schwarz and Jensen's inequalities to~\eqref{eq_boundsLoss}, i.e.,
\begin{align}
\Big( \!\nabla \! F(\pmb{m}_G(t)), \sum_{n=1}^{N} \!\nabla \! F(\pmb{m}^n_{i,j}(t)) \! \Big) \!=\! \| \nabla F(\pmb{m}_G(t)) \|^2,
\label{eq_CauchySchwarz}
\end{align}
\begin{align}
\left\|\sum_{n=1}^N \nabla F(\pmb{m}^n_{i,j}(t))\right\|^2 \leq \Big(\|\nabla F(\pmb{m}_G(t))\| + \sigma\Big)^2.
\label{eq_Jensen}
\end{align}
Let $\sigma$ denote a variance of the gradients due to non-IID data
\begin{align}
\sum^N_{n=1} \| \nabla F(\pmb{m}^n_{i,j}(t)) - \nabla \overline{F}(\pmb{m}^n_{i,j}(t)) \|^2 \leq \sigma^2,
\label{eq_variance}
\end{align}
where $\overline{F}(\pmb{m}^n_{i,j}(t))$ is the average of the model gradients.

By substituting~\eqref{eq_CauchySchwarz} and~\eqref{eq_Jensen} into~\eqref{eq_boundsLoss}, we have 
\begin{align}
F(\pmb{m}_G(t+1)) \leq &F(\pmb{m}_G(t)) - \mu_t \| \nabla F(\pmb{m}_G(t)) \|^2 + \nonumber \\
&\frac{L \mu_t^2}{2} (\| \nabla F(\pmb{m}_G(t)) \|+ \sigma)^2.
\label{eq_bound}
\end{align}

According to Cauchy-Schwarz Inequality, we have 
\begin{align}
F(\pmb{m}_G(t+1)) \leq &F(\pmb{m}_G(t)) + ({L \mu_t^2} - \mu_t)  \| \nabla F(\pmb{m}_G(t)) \|^2 \nonumber \\
& + {L \mu_t^2} \sigma^2.
\label{eq_bound 2}
\end{align}

By subtracting $F(\pmb{m}_G^*)$ from both sides of \eqref{eq_bound 2}, it follows 
\begin{align}
F(\pmb{m}_G(t+1)) - F(\pmb{m}_G^*) & \leq F(\pmb{m}_G(t))-  F(\pmb{m}_G^*) + {L \mu_t^2} \sigma^2 \nonumber \\ 
& + ({L \mu_t^2} - \mu_t)  \| \nabla F(\pmb{m}_G(t)) \|^2.
\label{eq_bound 3}
\end{align}

Considering Polyak-Lojasiewicz condition in Assumption~\ref{assumption}-1), we can obtain
\begin{align}
\left\| \nabla F(\pmb{m}_G(t+1)) \right\|^2 \leq & [2\rho ({L \mu_t^2} - \mu_t) + 1 ]\left\| \nabla F(\pmb{m}_G(t)) \right\|^2 \nonumber \\ 
& + 2\rho {L \mu_t^2} \sigma^2.
\label{eq_bound 4}
\end{align}

Based on the recurrence expression~\eqref{eq_bound 4}, we can obtain the upper bound of $\|\nabla F(\pmb{m}_G(t)) \|^2$, as given by
\begin{align}
&\left\| \nabla F(\pmb{m}_G(t)) \right\|^2 \leq [2\rho ({L \mu_t^2} - \mu_t) + 1 ]^t\left\| \nabla F(\pmb{m}_G(0)) \right\|^2 \nonumber \\ 
&~~~~~~~~ + 2\rho {L \mu_t^2} \sigma^2 \sum_{i=1}^t [2\rho ({L \mu_t^2} - \mu_t) + 1 ]^{t-1}.
\label{eq_bound 5}
\end{align}

Using the Polyak-Lojasiewicz condition again and Assumption~\ref{assumption}-2) provides
\begin{align}
F(\pmb{m}_G(t)) & - F(\pmb{m}_G^*)  
\leq [2\rho ({L \mu_t^2} - \mu_t) + 1 ]^t  \Theta \nonumber \\ 
& + 2\rho {L \mu_t^2} \sigma^2 \sum_{i=1}^t [2\rho ({L \mu_t^2} - \mu_t) + 1 ]^{i-1}.
\label{eq_bound 6}
\end{align}

According to~\eqref{eq_bound 6}, the convergence bound of HoVeFL in $T$ rounds can be given by
\begin{align}
F(\pmb{m}_G(T)) & -  F(\pmb{m}_G^*)  
\leq [2\rho ({L \mu_t^2} - \mu_t) + 1 ]^T  \Theta \nonumber \\ 
& + \frac{L \mu_t \sigma^2\{[2\rho ({L \mu_t^2} - \mu_t) + 1 ]^{T+1} -1\}}{L\mu_t - 1 }.
\label{eq_boundFinal}
\end{align}

To ensure the convergence of HoVeFL, the learning rate satisfies $\mu_t \leq \frac{1}{L}$ since $2\rho ({L \mu_t^2} - \mu_t) + 1 \leq 1$.
\begin{corollary}\label{rema}
The convergence upper bound is convex with respect to the number of communication rounds, i.e., $T$, if $\mu_t \leq \frac{1}{L}$ and $\Theta \geq \frac{(1- 2 \rho \mu_t + 2 \rho L \mu_t^2) L \mu_t \sigma_2^2}{1 -  L \mu_t}$. 	
\end{corollary}
\begin{proof}
Write the right-hand side of \eqref{eq_boundFinal} as a function of $T$, i.e., ${\cal{B}}(T)$. By deriving the second-order derivative of ${\cal B}(T)$ with respect to $T$ and letting $\frac{\partial^2 {\cal B}(T)}{\partial T^2} \geq 0$, we obtain $\mu_t \leq \frac{1}{L}$ and $\Theta \geq \frac{(1- 2 \rho \mu_t + 2 \rho L \mu_t^2) L \mu_t \sigma_2^2}{1 -  L \mu_t}$.
\end{proof}

It is important to address the difference between the update strategy in HoVeFL and the one in FedProx. While HoVeFL shares similarities in terms of employing a proximal term to stabilize the learning across non-IID data distributions, HoVeFL extends by integrating HFL and VFL. This integration is crucial in EdgeIoT, where device capabilities and data characteristics are diverse. Our convergence analysis demonstrates the theoretical underpinnings of our method.

\begin{figure}[htb]
\centering
\includegraphics[width=3.7in]{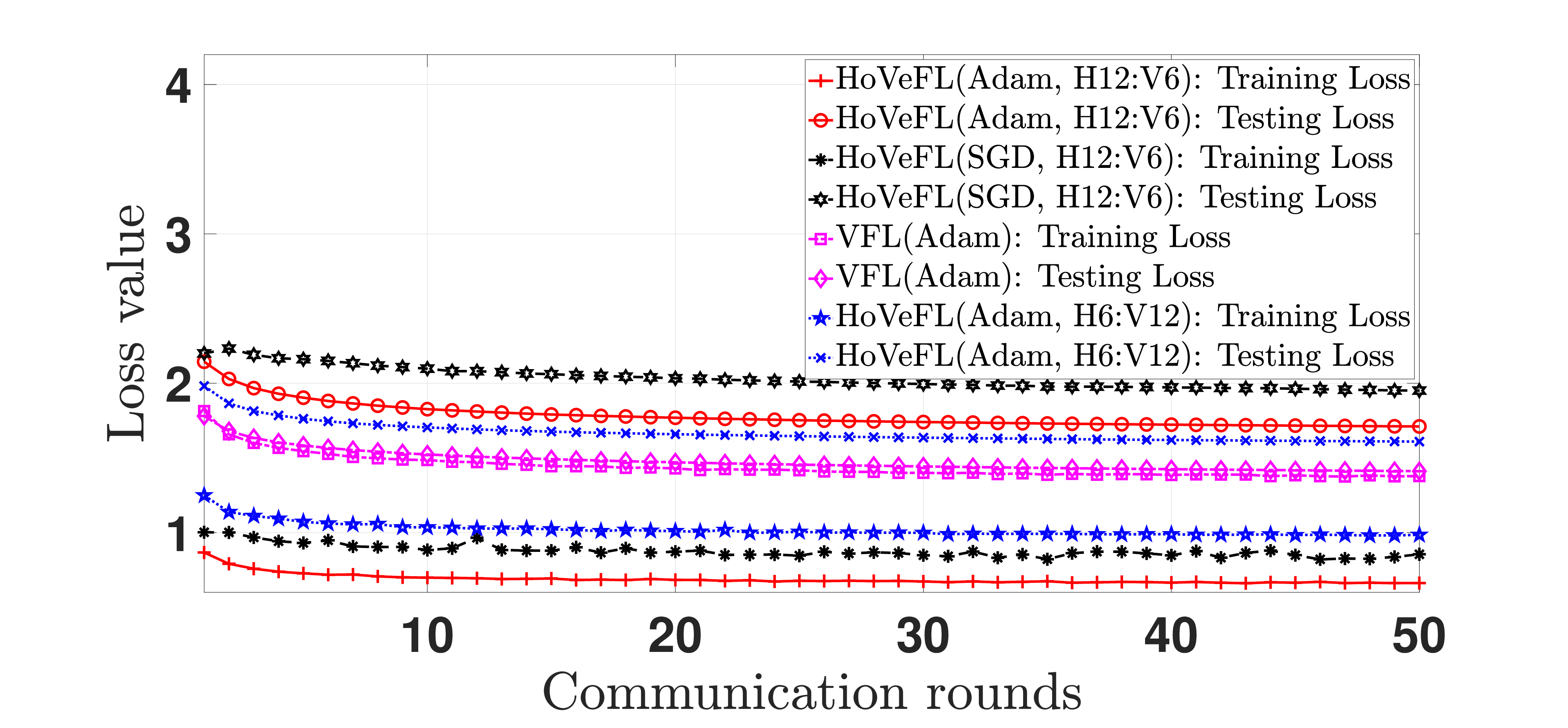}
\caption{The training and testing loss based on CIFAR-10. }
\label{fig_loss_cifar10}
\end{figure}

\begin{figure}[htb]
\centering
\includegraphics[width=3.7in]{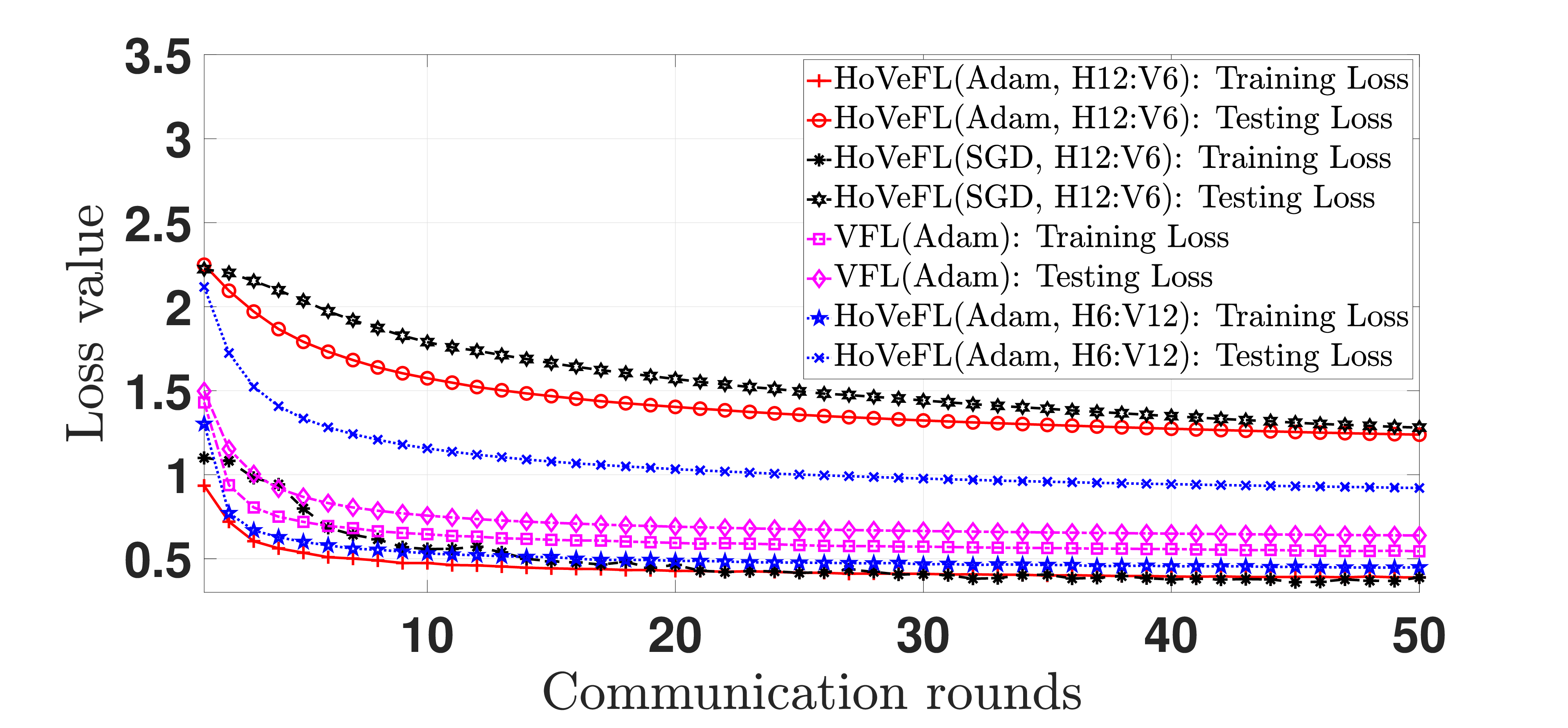}
\caption{The training and testing loss based on SVHN. }
\label{fig_loss_svhn}
\end{figure}

\section{Numerical Results}
In this section, we implement the proposed HoVeFL using PyTorch on a workstation with a GeForce RTX 2080 GPU.

Specific algorithmic steps of HoVeFL are as follows. 
\begin{itemize}
\item Local Data Processing and Feature Extraction: Each device processes its data locally using a Split Convolution Neural Network
with convolution layers to extract relevant features. For devices engaged in VFL, feature extraction maximizes the information gain from different data types; for HFL devices, the aim is to harmonize feature representation across non-IID datasets. 
\item Local Model Update: The local loss function in~\eqref{eq_lossFunc} is minimized, leading to an optimal local model ${\pmb{m}^{n*}_{i,j}(t)}$ at device $n$. Local model updates can be performed using a specified learning algorithm (e.g., SGD for handling data heterogeneity).
\item Model Aggregation: At the server, the model aggregation is conducted to combine vertical and horizontal updates. 
\item Global Model Synthesis and Distribution: After aggregation, a weighted model aggregation aggregates hybrid local models so that a global loss function $F(\pmb{m}_G(t))$ in~\eqref{eq_global_model} is minimized. The global model $\pmb{m}_G(t)$ is distributed across devices for further local training and refinement.
\end{itemize}

Figs.~\ref{fig_loss_cifar10} and~\ref{fig_loss_svhn} show the training and testing losses of HoVeFL on CIFAR-10 and street view house numbers (SVHN) datasets, respectively. Adam and SGD optimizers were tested on HoVeFL. We also adapted the numbers of devices performing HFL (i.e., $\bf{N_i}$) and devices performing VFL (i.e., $\bf{N_j}$). The performance of HoVeFL is compared with the one where all $N$ devices conduct VFL. 

Given 12 HFL devices and 6 VFL devices in HoVeFL, the training and testing losses with the Adam optimizer converges much faster than those with the SGD optimizer, confirming the convergence of Adam and SGD in~\cite{wilson2017marginal}. This validates the effectiveness of HoVeFL with heterogeneous HFL and VFL devices. Given the Adam optimizer, the testing loss of HoVeFL with 12 HFL devices and 6 VFL devices is 5.5\% and 25.2\% higher than the one with 6 HFL devices and 12 VFL devices on the CIFAR-10 and SVHN datasets, respectively. This is because the training dataset at the VFL devices ($N_j$) in HoVeFL is IID (due to training the local models on similar data samples $j$ but for different features) while the data at the HFL devices ($N_i$) in HoVeFL is non-IID (due to the same data feature $i$ but on different data samples). As more devices train IID data in HoVeFL, the model updates (parameters or gradients) are more homogeneous, resulting in more accurate updates toward minimizing the training and testing losses. 

Given the CIFAR-10 and SVHN datasets, the testing losses of HoVeFL (Adam, H6:H12) are 5.9\% and 33.3\% higher than those of VFL (Adam), respectively. This is reasonable since the devices carrying out HFL in HoVeFL hold data with the same feature space but distinct and potentially non-IID data samples. Moreover, the testing losses of HoVeFL (Adam, H12:H6) are 1.6\% and 2.5\% higher than those of HFL (Adam), respectively. This is because a combination of different data samples from VFL devices and different features from HFL devices means that the model must generalize across different characteristics and types of data. Handling non-IID data in such a diverse setup is inherently more difficult. The model may not optimize as effectively for any single type of data in HFL, leading to higher test losses due to the increased variance in model predictions.

In addition, Figs.~\ref{fig_loss_cifar10} and~\ref{fig_loss_svhn} show that the training and testing losses of HoVeFL have a significant difference, compared to VFL. This is because the training of HoVeFL includes a wide variety of data types and distributions, leading to difficulties in achieving a model that generalizes well across all types of data. The model can capture specific features or patterns present in the training data but not representative of the general population. Moreover, the ambitious goal of generalizing across different data features and different populations makes HoVeFL more susceptible to discrepancies in performance when faced with new or unseen data in the test set. On the other hand, VFL focuses on vertically integrating features for the same set of data samples. The diversity in data types is less pronounced in VFL than it is in HoVeFL. The alignment of features across the same set of the clients in VFL ensures that training and testing datasets are likely to follow similar distributions, thus maintaining a closer performance metric between training and testing.

\vspace{-5pt}
\section{Conclusions}
In this letter, we proposed a new HoVeFL for EdgeIoT, which enables some devices to train local models on different, independent features of the same data samples, while the other devices train local models for the same features using distinct, non-overlapping data samples. The training and testing losses were evaluated on the CIFAR-10 and SVHN datasets, with Adam and SGD. We also experimentally varied the number of devices performing HFL and VFL in HoVeFL. The performance was compared with a scenario where all EdgeIoT devices conduct VFL.

\ifCLASSOPTIONcaptionsoff
  \newpage
\fi

\vspace{-5pt}
\bibliographystyle{IEEEtran}
\bibliography{bibHVFL}

\begin{thebibliography}{10}
\providecommand{\url}[1]{#1}
\csname url@samestyle\endcsname
\providecommand{\newblock}{\relax}
\providecommand{\bibinfo}[2]{#2}
\providecommand{\BIBentrySTDinterwordspacing}{\spaceskip=0pt\relax}
\providecommand{\BIBentryALTinterwordstretchfactor}{4}
\providecommand{\BIBentryALTinterwordspacing}{\spaceskip=\fontdimen2\font plus
\BIBentryALTinterwordstretchfactor\fontdimen3\font minus
  \fontdimen4\font\relax}
\providecommand{\BIBforeignlanguage}[2]{{%
\expandafter\ifx\csname l@#1\endcsname\relax
\typeout{** WARNING: IEEEtran.bst: No hyphenation pattern has been}%
\typeout{** loaded for the language `#1'. Using the pattern for}%
\typeout{** the default language instead.}%
\else
\language=\csname l@#1\endcsname
\fi
#2}}
\providecommand{\BIBdecl}{\relax}
\BIBdecl

\bibitem{li2022internet}
K.~Li, Y.~Cui, W.~Li, T.~Lv, X.~Yuan, S.~Li, W.~Ni, M.~Simsek, and F.~Dressler,
  ``When internet of things meets metaverse: Convergence of physical and cyber
  worlds,'' \emph{IEEE Internet of Things Journal}, vol.~10, no.~5, pp.
  4148--4173, 2022.

\bibitem{ali2022federated}
M.~Ali, F.~Naeem, M.~Tariq, and G.~Kaddoum, ``Federated learning for privacy
  preservation in smart healthcare systems: A comprehensive survey,''
  \emph{IEEE Journal of Biomedical and Health Informatics}, vol.~27, no.~2, pp.
  778--789, 2022.

\bibitem{park2022federated}
J.~Park, J.~Moon, T.~Kim, P.~Wu, T.~Imbiriba, P.~Closas, and S.~Kim,
  ``Federated learning for indoor localization via model reliability with
  dropout,'' \emph{IEEE Communications Letters}, vol.~26, no.~7, pp.
  1553--1557, 2022.

\bibitem{liu2022vertical}
P.~Liu, G.~Zhu, W.~Jiang, W.~Luo, J.~Xu, and S.~Cui, ``Vertical federated edge
  learning with distributed integrated sensing and communication,'' \emph{IEEE
  Communications Letters}, vol.~26, no.~9, pp. 2091--2095, 2022.

\bibitem{li2023towards}
K.~Li, B.~P.~L. Lau, X.~Yuan, W.~Ni, M.~Guizani, and C.~Yuen, ``Towards
  ubiquitous semantic metaverse: Challenges, approaches, and opportunities,''
  \emph{IEEE Internet of Things Journal}, 2023.

\bibitem{zhang2023federated}
X.~Zhang, A.~Mavromatis, A.~Vafeas, R.~Nejabati, and D.~Simeonidou, ``Federated
  feature selection for horizontal federated learning in iot networks,''
  \emph{IEEE Internet of Things Journal}, vol.~10, no.~11, pp.
  10\,095--10\,112, 2023.

\bibitem{arunan2023federated}
A.~Arunan, Y.~Qin, X.~Li, and C.~Yuen, ``A federated learning-based industrial
  health prognostics for heterogeneous edge devices using matched feature
  extraction,'' \emph{IEEE Transactions on Automation Science and Engineering},
  2023.

\bibitem{you2022triple}
L.~You, S.~Liu, Y.~Chang, and C.~Yuen, ``A triple-step asynchronous federated
  learning mechanism for client activation, interaction optimization, and
  aggregation enhancement,'' \emph{IEEE Internet of Things Journal}, vol.~9,
  no.~23, pp. 24\,199--24\,211, 2022.

\bibitem{shi2023vertical}
Y.~Shi, S.~Xia, Y.~Zhou, Y.~Mao, C.~Jiang, and M.~Tao, ``Vertical federated
  learning over cloud-{RAN}: Convergence analysis and system optimization,''
  \emph{IEEE Transactions on Wireless Communications}, 2023.

\bibitem{su2021hierarchical}
L.~Su and V.~K. Lau, ``Hierarchical federated learning for hybrid data
  partitioning across multitype sensors,'' \emph{IEEE Internet of Things
  Journal}, vol.~8, no.~13, pp. 10\,922--10\,939, 2021.

\bibitem{yan2022latency}
Z.~Yan, D.~Li, X.~Yu, and Z.~Zhang, ``Latency-efficient wireless federated
  learning with quantization and scheduling,'' \emph{IEEE Communications
  Letters}, vol.~26, no.~11, pp. 2621--2625, 2022.

\bibitem{wang2020hybrid}
C.~Wang, J.~Liang, M.~Huang, B.~Bai, K.~Bai, and H.~Li, ``Hybrid differentially
  private federated learning on vertically partitioned data,'' \emph{arXiv
  preprint arXiv:2009.02763}, 2020.

\bibitem{xu2021fedv}
R.~Xu, N.~Baracaldo, Y.~Zhou, A.~Anwar, J.~Joshi, and H.~Ludwig, ``Fedv:
  Privacy-preserving federated learning over vertically partitioned data,'' in
  \emph{Proceedings of the 14th ACM workshop on artificial intelligence and
  security}, 2021, pp. 181--192.

\bibitem{lu2024two}
H.~Lu, W.~Chen, C.~Zhou, H.~Wu, F.~Lyu, and X.~S. Shen, ``A two-dimensional
  hybrid federated learning framework for secure data cooperation of multiple
  network service providers,'' \emph{IEEE Wireless Communications}, 2024.

\bibitem{thapa2022splitfed}
C.~Thapa, P.~C.~M. Arachchige, S.~Camtepe, and L.~Sun, ``Splitfed: When
  federated learning meets split learning,'' in \emph{Proceedings of the AAAI
  Conference on Artificial Intelligence}, vol.~36, no.~8, 2022, pp. 8485--8493.

\bibitem{liu2022wireless}
X.~Liu, Y.~Deng, and T.~Mahmoodi, ``Wireless distributed learning: A new hybrid
  split and federated learning approach,'' \emph{IEEE Transactions on Wireless
  Communications}, vol.~22, no.~4, pp. 2650--2665, 2022.

\bibitem{li2024data}
K.~Li, J.~Zheng, X.~Yuan, W.~Ni, O.~B. Akan, and H.~V. Poor, ``Data-agnostic
  model poisoning against federated learning: A graph autoencoder approach,''
  \emph{IEEE Transactions on Information Forensics and Security}, vol.~19, pp.
  3465--3480, 2024.

\bibitem{li2024leverage}
K.~Li, X.~Yuan, J.~Zheng, W.~Ni, F.~Dressler, and A.~Jamalipour, ``Leverage
  variational graph representation for model poisoning on federated learning,''
  \emph{IEEE Transactions on Neural Networks and Learning Systems}, 2024.

\bibitem{zheng2022exploring}
J.~Zheng, K.~Li, N.~Mhaisen, W.~Ni, E.~Tovar, and M.~Guizani, ``Exploring
  deep-reinforcement-learning-assisted federated learning for online resource
  allocation in privacy-preserving edgeiot,'' \emph{IEEE Internet of Things
  Journal}, vol.~9, no.~21, pp. 21\,099--21\,110, 2022.

\bibitem{karimi2016linear}
H.~Karimi, J.~Nutini, and M.~Schmidt, ``Linear convergence of gradient and
  proximal-gradient methods under the polyak-{\l}ojasiewicz condition,'' in
  \emph{Machine Learning and Knowledge Discovery in Databases: European
  Conference, ECML PKDD}.\hskip 1em plus 0.5em minus 0.4em\relax Springer,
  2016, pp. 795--811.

\bibitem{elgabli2022fednew}
A.~Elgabli, C.~B. Issaid, A.~S. Bedi, K.~Rajawat, M.~Bennis, and V.~Aggarwal,
  ``Fednew: A communication-efficient and privacy-preserving newton-type method
  for federated learning,'' in \emph{International Conference on Machine
  Learning (ICML)}.\hskip 1em plus 0.5em minus 0.4em\relax PMLR, 2022, pp.
  5861--5877.

\bibitem{wilson2017marginal}
A.~C. Wilson, R.~Roelofs, M.~Stern, N.~Srebro, and B.~Recht, ``The marginal
  value of adaptive gradient methods in machine learning,'' \emph{Advances in
  neural information processing systems}, vol.~30, 2017.

\end{thebibliography}

\end{document}